\pdfoutput=1
\RequirePackage{amsmath}
\documentclass[runningheads]{llncs}
\usepackage[legalpaper, margin=1.3in]{geometry}
\usepackage[textwidth=3.3cm]{todonotes}
\usepackage{amsmath,amssymb}
\usepackage[utf8x]{inputenc}
\usepackage{color}
\usepackage{algorithm}
\usepackage{algorithmicx}
\usepackage{algcompatible}
\usepackage{lastpage}

\usetikzlibrary{shapes}

\makeatletter
\def\moverlay{\mathpalette\mov@rlay}
\def\mov@rlay#1#2{\leavevmode\vtop{%
   \baselineskip\z@skip \lineskiplimit-\maxdimen
   \ialign{\hfil$\m@th#1##$\hfil\cr#2\crcr}}}
\newcommand{\charfusion}[3][\mathord]{
    #1{\ifx#1\mathop\vphantom{#2}\fi
        \mathpalette\mov@rlay{#2\cr#3}
      }
    \ifx#1\mathop\expandafter\displaylimits\fi}
\makeatother

\newcommand{\bigcupdot}{\charfusion[\mathop]{\bigcup}{\cdot}}

\algrenewcommand{\algorithmiccomment}[1]{\footnotesize {\fontfamily{ppl}\selectfont \% #1}}

\title{A Water-Filling Primal-Dual Algorithm for Approximating Non-Linear Covering Problems}%\thanks{This article was partially funded by Fondecyt Project Nr. 1181527.}}
\titlerunning{Approximating Non-Linear Covering Problems}
\author{Andrés Fielbaum\inst{1} 
\and Ignacio Morales\inst{2} \and José Verschae\inst{3}}

\institute{TU Delft, Delft, The Netherlands, \email{a.s.fielbaumschnitzler@tudelft.nl} \and
Universidad de O'Higgins, Rancagua, Chile,
\email{jose.verschae@uoh.cl} \and
Pontificia Universidad Católica, Santiago, Chile,
\email{inmorales@uc.cl}}

\def\NLKC{Non-Linear Knapsack-Cover}
\def\NLUFP{Non-Linear UFP-Cover}
\begin{document}
\maketitle

\section{Introduction}

\begin{abstract}
Obtaining strong linear relaxations of capacitated covering problems constitute a major technical challenge even for simple settings. For one of the most basic cases, the Knapsack-Cover (Min-Knapsack) problem, the relaxation based on \emph{knapsack-cover inequalities} achieves an integrality gap of 2. These inequalities have been exploited in more general environments, many of which admit primal-dual approximation algorithms.

Inspired by problems from power and transport systems, we introduce a new general setting in which items can be taken fractionally to cover a given demand. The cost incurred by an item is given by an arbitrary non-decreasing function of the chosen fraction. We generalize the knapsack-cover inequalities to this setting an use them to obtain a $(2+\varepsilon)$-approxi\-mate primal-dual algorithm. Our procedure has a natural interpretation as a bucket-filling algorithm, which effectively balances the difficulties given by having different slopes in the cost functions: when some superior portion of an item presents a low slope, it helps to increase the priority with which the inferior portions may be taken. We also present a rounding algorithm with an approximation guarantee of~2.

We generalize our algorithm to the Unsplittable Flow-Cover problem on a line, also for the setting where items can be taken fractionally. For this problem we obtain a  $(4+\varepsilon)$-approximation algorithm in polynomial time, almost matching the $4$-approximation known for the classical setting. 

% The abstract should briefly summarize the contents of the paper in
% 150--250 words.

\end{abstract}

%PROBLEM
% NAME: Non-Linear Knapsack Cover (with piecewise linear objective function)
% Cambiar: Part->segment
% Covering 

% Problem definition
% Applications
% Literature Review

% Main results

%	High level water-filling interpretation of Knapsack-cover primal dual.

Covering problems have been heavily studied by the combinatorial optimization community. Understanding their polyhedral descriptions, and how to approximate them, is a challenging and important task even for simple variants, as they often appear as sub-problems of more complicated formulations. In the literature of approximation algorithms, one of the main tools for obtaining strong linear relaxations of covering problems are the \emph{knapsack-cover inequalities}, introduced by Carr et al.~\cite{carr_strengthening_2000}. These inequalities and its generalizations have been heavily used in deriving approximation algorithms in different contexts~\cite{carnes_primal_dual_2015,cheung_primal_dual_2016,li_constant_2017,bansal_geometry_2014,mccormick_primal_dual_2017,bar-noy_unified_2001,bansal_geometry_2014,bansal_geometry_2019,moseley_scheduling_2019,levi_approximation_2008,an_lp-based_2017}. Although the family of inequalities is of exponential size, they can be approximately separated up to a factor of $1+\varepsilon$ in polynomial time. Additionally, in many cases they are well adjusted for primal-dual algorithms, which avoid having to solve the relaxation and yield combinatorial algorithms~\cite{bar-noy_unified_2001,carnes_primal_dual_2015,cheung_primal_dual_2016,mccormick_primal_dual_2017}.

In the classical Knapsack-Cover problem we are given a set of $n$ items $N$ and a demand $D\in\mathbb{N}$. Each item $i$ has an associated covering capacity~$u_i$ and cost~$c_i$. The objective is to choose a subset of items covering $D$ at a minimum cost. We introduce a new natural generalization of this problem motivated by applications on the operation of power systems and the design of public transport systems. In this version we can choose items partially at a given cost, which might be non-linear. More precisely, each item $i\in N$ have an associated non-decreasing cost function $f_i:\mathbb{N}\rightarrow \mathbb{Q}_{\ge0}\cup\{\infty\}$. We must choose a number $x_i\in\mathbb{N}$ for each $i$ such that $\sum_{i}x_i\ge D$. The cost of such solution is $\sum_{i\in N}f_i(x_i)$. There is a reduction of the Knapsack-Cover problem to this setting by considering $f_i(0)=0$, $f_i(1)=\ldots=f_i(u_i)=c_i$ and $f_i(x)=\infty$ for $x>u_i$. We say that we are in the \emph{list model} if the input contains the numbers $f_i(0),f_i(1),\ldots,f_i(D)$ explicitly as a list. In this case the reduction above is pseudo-polynomial. On the other hand, if each $f_i$ is given by an oracle that outputs $f_i(x)$ for any $x$, we say that we are in the \emph{oracle model}. In this case the reduction above can be made polynomial. %Hence, if each function $f_i$ is given by an oracle that allows to compute $f_i(x)$ in polynomial time, \NLKC{} is (weakly) NP-hard. On the other hand, if $f_i(0),f_i(1),\ldots$ is given as a list in the input, the problem can be solved in polynomial time with the classic  Depending how we describe the functions in the input 
We call our newly introduced problem the \emph{\NLKC{} problem}. Our setting also generalizes the \emph{Single-Demand Facility Location} problem studied by Carnes and Shmoys~\cite{carnes_primal_dual_2015}. In this setting each item (facility) has an activation cost $b_i$ and then the cost grows linearly at a rate of $a_i$, that is, $f_i(0)=0$, $f_i(x)= b_i + c_ix$ for $x\in\{1,\ldots,u_i\}$, and $f_i(x)=\infty$ otherwise.

More generally, we study the Non-Linear variant of the Unsplittable Flow-Cover problem on a path (UFP-cover), that extends the \NLKC{} problem. In the original UFP-cover problem, first considered by Bar-Noy et al.~\cite{bar-noy_unified_2001}, we have a discrete interval $I=\{1,...,k\}$ and a set $N$ of $n$ items, each one characterized by a capacity or height $u_i$, a cost $c_i$, and a sub-interval $I_i\subseteq \{1,\ldots,k\}$. We also have a demand $D_t$ for each $t \in I$. The problem consists on selecting the cheapest set of items such that the total height at any point in~$I$ is at least the demand, that is, we must pick a set $S$ minimizing $\sum_{i\in S}c_i$ such that $\sum_{i\in S:I_i\ni t} u_i\ge D_t$ for all $t\in I$. For this classic version, there is an algorithm that provides a 4-approximation based on the local-ratio framework~\cite{bar-noy_unified_2001}, or equivalently~\cite{bar_yehuda_equivalence_2005}, based on the primal-dual framework using knapsack-cover inequalities~\cite{cheung_primal_dual_2016}. In this paper we generalize this problem to the case where items can be taken partially. As before we are giving a non-decreasing function $f_i=\mathbb{N}\rightarrow \mathbb{Q}_{\ge 0}\cup\{\infty\}$ for each item. We can choose to set the height of any item to a value $x_i\in \mathbb{N}$ by paying a cost $f_i(x_i)$. We mush choose heights in order to cover the demand at each point $t\in I$ at a minimum total cost $\sum_{i}f_i(x_i)$. Notice that this setting generalizes \NLKC{}. 
% We will show that extending the ideas of Section~\ref{Sec KP} we can also achieve a 4-approximation for the non-linear case. Notice that without loss of generality we can assume $k\le 2n$ as a point in $I$ which is not an end-point of some $I_t$ can be easily removed.

% More generally, we consider the Unsplittable Flow ... In this setting.... The best approximation algorithm for this setting is a $4$-approximation based on the primal-dual framework, using knapsack-cover inequalities~\cite{}.
In this article we provide a generalization of the knapsack-cover inequalities to the \NLKC{} problem, which we also apply to \NLUFP{}. %For both problems, 
The obtained relaxations yield primal-dual algorithms matching the classical settings. Namely, for \NLKC{} we show a  2-approximation algorithm, and for \NLUFP{} a 4-approximation algorithm, both running in polynomial time in the list model. They can be adapted to yield a $(2+\varepsilon)$- and $(4+\varepsilon)$-approximation, respectively, in polynomial time for the oracle model. Additionally, we show a rounding technique for the \NLKC{} case also achieving a $2$-approximation for the list model. %the same approximation ratio, and provide a (pseudo)polynomial separation oracle\jvcom{To be done}.

\noindent\emph{Motivation.} One of our main motivations for considering non-linear cost functions comes from the Unit Commitment Problem (UCP), a prominent problem in the operation of power systems. In its most basic version, a central planner, called Independent System Operator (ISO), must schedule the production of energy generated from a given set of power plants, in order to satisfy a given demand. A common issue in this setting is that plants incur fixed costs for starting production, and after the resource is available, a minimum amount of energy must be produced. For the case of one time period, the problem corresponds exactly to \NLKC{}. It is worth noticing that after paying the fixed cost the behavior of the cost functions might be non-linear and are often modelled by convex quadratic functions~\cite{zhu_optimization_2015}.

On the other hand, \NLUFP{} appears in the optimization of transport systems. Consider a long avenue in which there are several bus stops $\{1,\ldots,k\}$, and passengers need to move (in a single direction) within them. This yields some demand $D_t$ at each point $t$, representing the total number of passengers that contain $i$ inside their path. On the supply side, there are potential bus transit lines, each covering some sub-interval of the avenue, which could be longer routes that intersect the avenue in and out in given points in $\{1,\ldots,k\}$. Each line can supply different capacities through some optimal combination of frequencies and bus sizes, represented by line-specific cost functions. There is a vast literature concerning economies of scale in public transport lines: for instance, Mohring~\cite{Mohring72optimisationand} states that there are economies of scale in public transport, Fielbaum et al.~\cite{fielbaum_beyond_2019} show that they get exhausted, while Coulombel and Monchambert~\cite{Coulombel_Diseconomies_2019} propose that the system could face diseconomies of scale when the demand exceeds certain thresholds. Hence, techniques to manage non-linear functions (that can have convex and concave regions) are needed.

\noindent\emph{Related Work.}  
The use of the primal-dual method to derive approximation algorithms was introduced by Bar-Yehuda and Even~\cite{bar-yehuda_lineartime_1981} and Chvátal~\cite{chvatal_greedy_1979} and then became an important general tool for designing approximation algorithms~\cite{williamson_design_2011}.

% Carr et al.~\cite{carr_strengthening_2000} studied the Knapsack-Cover problem from a polyhedral perspective and derived a set of inequalities called \emph{knapsack-cover inequalities} that give an integrality gap of 2 when applied over the straightforward  linear programming relaxation. This set of inequalities have proven very useful in a myriad of different covering problems. 

The first work to consider the primal-dual setting based on knapsack-cover inequalities were by Bar-Noy et al.~\cite{bar-noy_unified_2001}. However, they posed their algorithm in the equivalent local-ratio framework~\cite{bar_yehuda_equivalence_2005}, even before the knapsack-cover inequalities were introduced and without stating the underlying LP-relaxation. Their techniques yield a 4-approximation algorithm and their analysis is tight~\cite{cheung_primal_dual_2016}. Additionally, this problem admits a quasi-polynomial time approximation scheme (QPTAS)~\cite{hohn_how_2014}. On the other hand, Carnes and Shmoys~\cite{carnes_primal_dual_2015} gave an explicit description of the primal-dual method, obtaining a 2-approximation for Knapsack-Cover, the Single-Demand Facility Location problem, and the more general Single-item Lot-Sizing problem with Linear Holding Costs.   %The latter setting, which generalizes the previous two, considers the lot-sizing problem on a finite period. In each period a given demand must be satisfied, and we should decide on setting an order (at a fixed cost) of a given number of products whose cost grows linearly. Inventory can be saved to later periods, also at a linear costs. 
 Cheung et al.~\cite{cheung_primal_dual_2016} consider the Generalized Min-Sum Scheduling problem on a single machine without release dates; they obtain a $(4+\varepsilon)$-approximation algorithm based on the primal-dual framework on an LP with knapsack-cover inequalities. Finally, McCormick et al.~\cite{mccormick_primal_dual_2017} consider covering problems with precedence constraints, where they are able to give a primal-dual algorithm with approximation ratio equal to the width of the precedence relations. We remark that \NLKC{} can be modeled within this framework, but applying this result yields an unbounded approximation guarantee in this case.
 
Outside the primal-dual framework, there is also a rich literature on the use of the knapsack-cover inequalities and its generalizations together with rounding techniques. The problems considered include the Min-Sum General Scheduling problem on a single~\cite{bansal_geometry_2014} and multiple~\cite{moseley_scheduling_2019,bansal_geometry_2019} machines, the Uniform~\cite{levi_approximation_2008} and Non-Uniform~\cite{li_constant_2017} Capacitated Multi-item Lot-sizing problem, and Capacitated Facility Location~\cite{an_lp-based_2017}. On the other hand, it is not known if there exists a compact set of constraints matching the strength of the knapsack-cover inequalities. Recently, Bazzi et al.~\cite{bazzi_small_2018} gave a formulation with an integrality gap of $2+\varepsilon$ for the Knapsack-Cover problem with a quasipolinomial number of inequalities.

It is also worth mentioning that the most common technique for dealing with non-linear cost functions in capacitated covering problems is a doubling technique: split the cost function in segments where the function doubles. Then, each segment can be considered independently as a single item. This removes the precedence dependence between different segments, at a cost of losing a factor 4 in the approximation ratio; see for example~\cite{bansal_geometry_2014}. Our approach strengthen the knapsack-cover inequalities and allows to avoid the extra factor lost. 

%Carnes and Shmoys \cite{carnes_primal_dual_2015} consider a similar setting to ours, with the difference that the $p_i$'s are all zero, and derive a generalization of the KC-inequalities for this setting, which then apply to the more general capacitated single-item lot-sizing problem. They obtain an elegant 2-approximation algorithm for this setting via a primal-dual schema.

\noindent\emph{Our Contribution.} %To appreciate our algorithmic contribution, consider the exercise of deriving a greedy algorithm for the \NLKC{} problem. Let $z_{ij}$ be a binary variable that represents whether the $j$-th unitary \emph{segment} of item $i$ is \emph{taken} as part of the solution. That is, for each item we are have that $x_i=\sum_j z_{ij}$ units to cover the demand, where it must hold that if $z_{ij}=1$ then $z_{i,j-1}=1$ for all $i,j$. With this, the cost of the solution is $\sum_i f_i(x_i)=\sum_{ij}g_{ij}z_{ij}$ for $g_{ij}=f_i(j)-f_i(j-1)$. In a greedy algorithm, one might be tempted to take segments with low $g_{ij}$. This fails as such segments might be preceded with another segment with large cost $g_{ik}$ for $k< j$ (notice that this does not happen for convex $f_i$, which is a considerably simpler case). This poses two fundamental questions when assessing the value of a segment: (i) how to take into consideration previous segments of high cost? (ii) how to take into account low costs of segments to the right, specially considering segments that finally might not be part of the final solution (since the demand can be covered by other items)?\jvcom{Son realmente dos preguntas o son en realidad la misma?} \afcom{Me parece bien en general el parrafo, pero creo que es un poco redundante al comienzo. Dejo una propuesta alternativa abajo en el texto fuente para ir convergiendo; creo que esta bien decir que son dos preguntas, pues para cada segmento ambas cosas deben ser tomadas en cuenta al mismo tiempo, y hay diferencias: lo que esta a tu izquierda es obligatoria, lo que esta a tu derecha puede ser tomado o no ser tomado.}
%Propuesta de redaccion alternativa:
%Consider the exercise of deriving a greedy algorithm for the \NLKC{} problem. 
Let $z_{ij}$ be a binary variable that represents whether $x_i \geq j$ in the solution, i.e., if the $j$-th unitary \emph{segment} of item $i$ is \emph{taken}. Defining $g_{ij}=f_i(j)-f_i(j-1)$, then the cost of any solution is $\sum_{ij}g_{ij}z_{ij}$, and for any solution to be feasible it must hold for all~$i,j$ that if $z_{ij}=1$ then $z_{i,j-1}=1$ . In a greedy algorithm, one might be tempted to take segments with low $g_{ij}$. This fails as such segments might be preceded by another segment with $g_{ik}\gg g_{ij}$ for $k< j$. % (notice that this does not happen for convex $f_i$, which is a considerably simpler case). 
This poses two fundamental questions when assessing the value of a segment: (i) how to take into consideration (mandatory) preceding segments of high cost? (ii) how to take into account low costs segments to the right, specially considering segments that finally might not be part of the final solution (since the demand can be completely covered by previous segments)?

We introduce a natural variant of the knapsack-cover inequalities for non-linear cost functions. These generalize the basic version of the inequalities, as well the generalization of Carnes and Shmoys~\cite{carnes_primal_dual_2015} for the Single-Demand Facility Location Problem. Our inequalities are then used to derive a primal-dual algorithm that helps to handle the fundamental questions above. Our algorithm can be interpreted as a water-filling algorithm. %Each item has a corresponding stairs of buckets. 
Each segment~$j$ of an item~$i$ has a corresponding \emph{bucket} $B_{ij}$ of capacity $g_{ij}$, representing an inequality in the dual linear program. %Notice then that there is a relationship between low-cost segments and buckets that get full quickly. 
All buckets for a given item $i$ are placed on a stairway, where bucket $B_{ij}$ is on the $j$-th step of the stairs. A segment is taken, i.e. we set $z_{ij}=1$, if its corresponding bucket and all previous ones (which are in lower steps of the stairs) are full. Water reaches buckets through two mechanisms. Water from an external source is poured directly into each bucket at a rate of either 1 or 0 (units of water per time unit). The first time a bucket $B_{ij}$ becomes full, then the water arriving to this bucket spills to bucket $B_{i,j-1}$, which now fills at a rate of 2 (as long as $B_{ij}$ is still receiving water from the external source). If $B_{i,j-1}$ also becomes full and $j>2$, then the water pouring into $B_{ij}$ and $B_{i,j-1}$ spills to $B_{i,j-2}$ which now fills at a rate of 3, etc. For a bucket to receive water from the external source it must satisfy two properties: (i) its corresponding segment has not been taken yet in the primal solution, and (ii) all previous segments of the item are not enough to cover the remaining demand. % \jvcom{Cambié esta última frase, espero que se entienda mejor} %This process continues until all buckets $B_{i1},\ldots,B_{ip}$ are filled for some $p$. At this moment we can take all segments $1,\ldots, p$ of item $i$ in our primal solution, that is, we set $z_{i1}=\ldots,z_{ip}=1$. 
Our primal-dual algorithm helps to take care of the tensions implied by the questions above by making buckets filling faster due to water spilled from higher buckets, and prevents spilling water from a bucket if they are so high that they are useless to help covering the remaining demand. %as well as stopping the flow from the external source on buckets that are too high the ladder and will not help covering the demand. 

For the case of \NLUFP, our algorithm works similarly. However, the primal solution constructed with the algorithm can contain redundant segments due to sub-intervals of $I$ that can be covered in subsequent steps of the algorithm. For this reason we need to perform a \emph{reverse-delete} (or pruning) strategy to remove unnecessary segments, in the reverse order in which they were introduced in the primal solution. 

We notice that, for both algorithms, our analysis is tight as they achieve the same performance guarantee as their classic variants~\cite{carnes_primal_dual_2015,cheung_primal_dual_2016}. Additionally, the integrality gap of our formulation for the \NLKC{} problem is also 2, as the same lower bound of the the classical setting holds~\cite{carr_strengthening_2000}. 

Finally, we show an iterative rounding technique for the LP relaxation of the \NLKC{} problem. The rounding first round up variables larger than $1/2$ to 1, and then considers the residual problem and a simplified LP relaxation. This relaxation has only few inequalities and an extreme point can be easily rounded as we can show most of its variables are integral.

\section{A Generalization of the Knapsack-Cover Inequalities for Non-Linear Knapsack-Cover} \label{Sec KP}

%  In the next section we generalized out techniques to \NLUFP{}.

% \subsection{Basic Definitions}

We first study the \NLKC{} problem. Recall that in this setting we consider a set $N$ of $n$ items, each with a non-decreasing function $f_i$ %$f_i:\mathbb{N}\rightarrow \mathbb{Q}_{\ge 0}\cup\{\infty\}$,
and a demand $D\in\mathbb{N}$. We assume that all functions $f_i$'s are  defined over a common domain $\{0,1,\ldots,m\}$, for some $m\le D$. Hence, each item $i\in N$ has $m$ \emph{segments} of unit length, indexed by a common set $M=\{1,\ldots, m\}$, each having a unit cost $g_{ij}=f_i(j)-f_i(j-1) \geq 0$. In what follows we assume that our instance admits a feasible solution. %Also, we will say that the $j$-th segment of item $i$ has been \emph{taken} into the solution if $z_{ij}=1$.
We start by considering the list model.%First we consider the list model, where each function $f_i$ is given as a list of $m+1$ numbers. %In the full version we extend this result to the oracle model and obtain a polynomial time algorithm for this setting at a loss of a $1+\varepsilon$ factor in the approximation guarantee.
%  Con sider a model with variables $z_{ij}$ that denotes whethear the amount of demand $D\in \mathbb{N}$ covered by the $j$-th segment of item $i$ (in an optimal solution, $z_{ij}$ is going to be binary). Hence, $\sum_{i,j}g_{ij}z_{ij}$ is the total cost of the solution, and the constraint $\sum_{i\in N}\sum_{j\in M} z_{ij}\ge D$ guarantees that the demand is covered. Additionally, to guarantee that the solution is feasible we must enforce that if $z_{ij}>0$ then $z_{ik}=1$ for all $k\in\{1\ldots,j-1\}$. 

It is worth mentioning that the problem described can be solved in polynomial time (respectively pseudo-polynomial) in the list (respectively oracle) model by a straightforward adaptation of the classical dynamic program for Knapsack. In the oracle model the problem is (weakly) NP-hard as it contains Knapsack-Cover as a special case. For this model the dynamic program can be turned into an FPTAS also by adapting well known rounding  techniques~\cite{williamson_design_2011}. However, these techniques cannot handle \NLUFP{}.

\subsection{Knapsack-Cover Inequalities for Non-Linear Costs}
%In order to write a linear relaxation of this problem, consider the following thought experiment. 
To write a linear relaxation of this problem, consider $a\in \{0,\ldots,m\}^N$, where $a_i$ represents that all segments $j\in\{1,\ldots, a_i\}$ have been taken already for item $i\in N$ (and $a_i=0$ represents that no segment of $i$ is taken yet). We face the residual problem, where we must decide about segments not taken yet, and we must cover the residual demand
% \begin{equation}
 $D(a):=\max \{D-\sum_{i\in N}a_i,0\}$.
% \end{equation}
%where $x_+=\max\{x,0\}$.

We now relax the condition that $z_{ij}=1$ implies $z_{ik}=1$ for $k<j$. To do so, note that in a feasible solution variable $z_{ij}$ should never be larger than $\min\left\lbrace z_{ij},z_{i,j-1},\ldots,z_{i1} \right\rbrace$, and hence we can replace in our formulation the appearance of $z_{ij}$ by this minimum. Additionally, in order to cover the residual demand $D(a)$, an optimal feasible solution will never set a variable $z_{ij}$ to 1 if $j>a_i+D(a)$. Hence, for item $i$ we can only take up to segment $m_i(a):= \min\{m,a_i+D(a)\}$. We conclude that the following is a relaxation of the \NLKC{} problem, which we call [GKC]:
\begin{align}
\nonumber
 & \min \sum_{i \in N,j\in M} g_{ij}z_{ij}\\
\label{eq:GKC-Convex}
& \sum_{i\in N}\sum_{j= a_i+1}^{m_i(a)} \min\left\lbrace z_{ij},z_{i,j-1},\ldots,z_{i1} \right\rbrace \ge D(a) &\text{for all }a\in \{0,\ldots,m\}^N,\\
\nonumber
&z_{ij}\ge 0 & \text{for all }i\in N, j\in M.
\end{align}
We call the set of inequalities~\eqref{eq:GKC-Convex} the \emph{knapsack-cover inequalities for non-linear costs}. %It is not hard to see that this relaxation is convex. 
This relaxation can be easily linearized. Indeed, if a program has a constraint of the form $\min\{x_1,x_2\}\ge b$, then we can replace it with $x_1\ge b$ and $x_2\ge b$. More generally, if the constraint is $\min\{x_1,x_2\} + \min\{x_3,x_4\}\ge b$, then we %obtain an equivalent linear program by 
must consider all constraints $x_i+x_j\ge b$ for all $i\in\{1,2\}$ and $j\in\{3,4\}$. %Applying this idea to our case, 
Here, we can replace each convex inequality in \eqref{eq:GKC-Convex} with (exponentially) many linear ones. The linear inequalities can be constructed by replacing each summand $\min\left\lbrace z_{ij},z_{i,j-1},\ldots,z_{i1} \right\rbrace$ in \eqref{eq:GKC-Convex} by one of its terms $z_{ij},z_{i,j-1},\ldots,z_{i1}$. Each new linear inequality will be indexed by a pair $(a,F)$, with $F=(F^k_i)_{i\in N,k\in\{0,...,m-1\}}$. Each $F^k_i$ is a set of indices, such that $z_{i,j-k}$ is chosen from $\min \{z_{ij},z_{i,j-1},\ldots,z_{i1} \}$ iff $j \in F^k_i$. We also require that $(F_i^k)_{k=0}^{m-1}$ is a partition of $\{a_i+1,\ldots,m\}$, and $j\in F_i^k$ implies that $j-k\ge1$. Then, an inequality in~\eqref{eq:GKC-Convex} indexed by a given vector $a$ can be replaced with the following set of constraints:
\begin{equation}
\label{eq:linearized}
\sum_{i\in N}\sum_{k=0}^{m-1} \sum_{j\in F_i^k:j\le m_i(a)} z_{i,j-k}\ge D(a) \qquad \text{for all }F. 
\end{equation}
Let us consider now a given pair $(a,F)$, an item $i$, and $j>a_i$. The term $z_{ij}$ might appear several times in the respective constraint~\eqref{eq:linearized}, depending on how many ``minimums'' are replaced by it. If $k \in F_i^{k-j}$, then $\min \{z_{ik},z_{i,k-1},\ldots,z_{i1} \}$ is replaced by $z_{ij}$. %Hence, $z_{ij}$ appears $|\{k\ge j: k \in F^{k-j}_i\}|$ times (note that this quantity might be zero if $j \notin F_i^0$).
Moreover, recall that the residual demand $D(a)$ will never be covered by a segment $z_{ij}$ for $j>m_i(a)$, and hence the number of times that $z_{ij}$ appears in the left-hand-side of the inequality is
\begin{equation}
    \tau(i,j,a,F)=\lvert\{k \geq j: k\in F_i^{k-j}, k \leq m_i(a)\} \rvert.
\end{equation}
With this, we obtain the following relaxation, which is equivalent to [GKC],

\begin{align*}
%  \text{[P-GKC]:} \min &\sum_{i\in N}\sum_{j\in M} z_{ij}g_{ij}\\
% \text{s.t. } & \sum_{i\in N}
% \sum_{j=0}^{m-1} \geq D(a) &  \text{for all } (a,F) \in \mathcal{F}, \label{rest:LP_3}\\
% & 0\le z_{ij} & \text{for all } i \in N\ \text{for all } j\in M,
 \text{[P-GKC]:} \min &\sum_{i\in N}\sum_{j\in M} z_{ij}g_{ij}\\
\text{s.t. } & \sum_{i\in N}
\sum_{j=a_i+1}^{m} \tau(i,j,a,F) \cdot z_{ij}\geq D(a) &  \text{for all } (a,F) \in \mathcal{F}, \label{rest:LP_3}\\
&z \geq 0 & 
\end{align*}
where
\[
 \mathcal{F}=\left\lbrace (a,F):\substack{\displaystyle a\in \{0,\ldots,m\}^n,\, \bigcupdot_{k=0}^{m-1}F_i^k= \{a_i+1,\ldots,m\},\\\displaystyle\text{ and if }  j\in F_i^k \text{ then } j-k\ge 1 \text{ for all } k, j} 
\right\rbrace.
\]
A routinary computation yields that the dual of this linear program is as follows.
\begin{align}
\nonumber \text{[D-GKC]:} \max&\sum_{(
a,F)\in \mathcal{F}}D(a)v_{aF}\\
\text{s.t. } & \sum_{(a,F)\in \mathcal{F}: a_i<j} \tau(i,j,a,F)\cdot v_{aF} \le g_{ij} & \text{for all } i\in N, j \in M, \label{rest:DP_3}\\
\nonumber &  v\geq 0. & 
\end{align}
\vspace{-1cm}
\subsection{A 2-approximate Primal-Dual Algorithm}\label{sec:Primal-DualKC}

We provide a primal-dual 2-approximation algorithm based on the LP-relaxation [P-GKC] and its dual [D-GKC]. It is worth having in mind the bucket representation of the algorithm given above in the introduction. %We first give a high level and intuitive description of our algorithm, with the objective of making the technical part easier to follow. 

\noindent\textbf{Algorithm description.} 
 The water-filling algorithm described in the introduction is an intuitive representation of a greedy algorithm for the dual [D-GKC]. Each bucket has a correspondence to a dual inequality: Each of the inequalities in the dual [D-GKC] represents a bucket, the left hand size corresponds to the amount of water in the bucket, while the right hand side is its capacity. The greedy dual algorithm raises dual variables one by one, starting from a dual solution $v\equiv 0$, implying the increase on water on the buckets. In each iteration of the main loop, we raise a variable $v_{aF}$. The index $a$ is chosen such that $a_i$ represents the largest value $\ell$ for which all buckets $B_{i1},\ldots,B_{i\ell}$ are full (or equivalently, $z_{i1}=\ldots=z_{i\ell}=1$), for each $i\in N$. To choose $F$, a segment $j$ will belong to $F_i^k$ if and only if the water from the external source falling into bucket $B_{ij}$ (if any) spills down to bucket $B_{i,j-k}$. Number $k$ is chosen such that it is the smallest number for which $B_{i,j-k}$ is not full, representing the idea that the water of full buckets falls down to the previous buckets on the stairs. Also, buckets receiving water from the external source are the buckets $B_{ij}$ with $a_i+1\le j\le m_i(a)$. This way, $\tau(i,j,a,F)$ corresponds to the filling rates of bucket $B_{ij}$ in the current iteration, which considers the water directly from the external source and the water spilled from higher buckets. We stop raising variable $v_{aF}$ as soon as one dual inequality becomes tight, i.e., some bucket $B_{ij}$ becomes full. After, we update the value of $z$ by setting $z_{ik}=1$ if $B_{i\ell}$ is full for all $\ell\le k$. Remark that there we do not require that $k\le m_i(a)$, and hence the returned primal solution might not satisfy the total demand exactly. Finally, we update $a$ and $F$ as described above and repeat the main loop until the residual demand $D(a)$ reaches $0$. The precise definition in pseudocode is given in Algorithm~\ref{alg:Primal-DualKC}, which uses Algorithm~\ref{alg:Update} as a sub-routine in Line 14. %Due to space restrictions we omit the pseudocode, however, the precise definition can be obtained by considering Lines 1-3, 5, 7-12, and 14-15 from Algorithm~1, omitting any reference to index $t$, and constructing the output by setting $z_{ij}=1$ iff $j\le a_i$ for all $i,j$.
\begin{algorithm}
\caption{Primal-Dual Water-Filling Algorithm for the Knapsack-Cover Problem with Non-Linear Costs.}\label{alg:Primal-DualKC}
\begin{algorithmic}[1]
\STATE $z$, $v$ $\leftarrow$ 0; \Comment{primal and dual solutions.}
\STATE $a \leftarrow 0$; \Comment{$a_i$ represents the largest value for which buckets $B_{i1},\ldots,B_{i,a_i}$ are full.}
\STATE $F_i^k\leftarrow \emptyset \ \ \forall i\in N,k\in M$;
\STATE $F_i^0 \leftarrow M$; \Comment{$j\in F_i^k$ iff  water from bucket $B_{ij}$ falls to bucket $B_{i,j-k}$.}

\WHILE{ $D(a)>0$} \label{ln:MainLoopKC}
\STATE Increase $v_{aF}$ until a dual constraint indexed by $(i,j)$, for some item $i\in N$ and segment $j\in M$, becomes tight. %Break ties in favor of smaller $j$.
\Comment{Bucket $B_{ij}$ becomes full.} 

\Comment{Update $z_{ij}$:}
\IF{$j=a_i+1$}\label{ln:IF}

	\STATE Let $q>a_i$ be the maximum number such that $B_{i,a_i+1},\ldots,B_{iq}$ are full.\label{line:k}
	\FOR{$\ell=j,\ldots,q$}
%       \Comment{If demand is still not satisfied}
		\STATE $z_{i\ell} \leftarrow 1$;	\Comment{Take available segments.}
% 		\STATE quitar $j''$ del $F_i^k$ en el que esté;
% 		\IF{$D(a)>0$}	
% 		\RETURN $z,v$;%terminar \textbf{for};
% 		\ENDIF
	\ENDFOR
	\STATE $a_i \leftarrow q$;
\ELSE \Comment{If we cannot raise variable $z_{ij}$:}
    \STATE $F \leftarrow  \text{Update}(F,i,j,a)$ \Comment{Call Algorithm 2 to update the sets $F$}
% 	\STATE Mover $j$ de $F_i^0$ a $F_i^{(j-j')}$;
\ENDIF\label{ln:endIF}
\ENDWHILE
\STATE \textbf{return} $v,z$.
% \RETURN $(\bar{a},\bar{F},\bar{z}) \leftarrow (a,F,z)$;
\end{algorithmic}
\end{algorithm}

\begin{algorithm}
\caption{$\text{Update}(F,i,j,a)$: Updating Buckets Subroutine}
\label{alg:Update}
\begin{algorithmic}
\STATE \textbf{input} $a,i,j,F$ \Comment{$a,F$ represent the current state of the buckets, $i,j$ represent which is the bucket that just got full. We require  $j>a_i+1$}.
\STATE Let $p<j$ be the maximum number so that $B_{ip}$ is not full.
\STATE Let $q>j$ be the minimum number so that $B_{iq}$ is not full (and $q=m+1$ if $B_{ij},\ldots,B_{im}$ are all full).
\FOR{$\ell=j,j+1,\ldots,q-1$}
		\STATE $F_i^{\ell-j} \leftarrow F_i^{\ell-j}\setminus \{\ell\}$ and  $F_i^{\ell-p} \leftarrow F_i^{\ell-p}\cup \{\ell\}$ \Comment{The water from the external source falling to $B_{i\ell}$, which was previously spilling to bucket $B_{ij}$, now spills to bucket $B_{ip}$.}
\ENDFOR 
\STATE \textbf{return} $F$
\end{algorithmic}
\end{algorithm}
\noindent\textbf{Analysis.} %It is not hard to see that the algorithm terminates. Indeed, 
The algorithm terminates, as each iteration of the main loop (Line~5) corresponds to some bucket that becomes full,
%\footnote{It is worth noticing that in some situations $v_{af}$ might not increase in Line~\ref{ln:MainLoopKC}. This can happen if two ore more buckets became full in a previous iteration simultaneously, and hence one of those buckets $B_{ij}$ was not processed in Lines~\ref{ln:IF}-\ref{ln:endIF}. In this case the algorithm can now take such bucket $B_{ij}$ and process it in the current iteration as it were the bucket becoming full.}, 
so we enter the while loop at most $nm$ times. Therefore the algorithm runs in polynomial time in the list model. %Similarly, in each iteration we maintain a feasible dual solution $v$ and the algorithm only stops when the primal $z$ is feasible. 
The main challenge is to show that the algorithm is 2-approximate.

It follows directly that the dual solution constructed is kept feasible through the execution of the algorithm. The fact that the primal solution is feasible follows since we kept iterating the main loop until the residual demand $D(a)$ is zero. As in most approximate primal-dual algorithms, the crux of the analysis is to show that an approximate form of the complementary slackness conditions are satisfied. This is summarized in the next lemma.

\begin{lemma}\label{lm:AproxCompSlackness} Let $\bar{v},\bar{z}$ be the primal and dual solutions computed by the algorithm. Then it holds that 
\[
 \sum_{i\in N}\sum_{j= a_i+1}^m \tau(i,j,a,F)\bar{z}_{ij} \le 2D(a),
\]
for all $(a,F)\in \mathcal{F}$ such that $\bar{v}_{aF}>0$.
\end{lemma}

Before showing this lemma, let us show how to use it for the proof of the main theorem.

\begin{theorem}\label{thm:PrimalDual2}
Algorithm~\ref{alg:Primal-DualKC} is a 2-approximation for the Knapsack-Cover Problem with Non-Linear Costs.
\end{theorem}

\begin{proof}
 Let $\bar{v},\bar{z}$ be the primal and dual solutions computed by the algorithm. The cost of our solution is
 \begin{align*}
  \sum_{i\in N}\sum_{j\in M} \bar{z}_{ij} g_{ij}
&=\sum_{i\in N}\sum_{j\in M} \bar{z}_{ij}\left(\sum_{(a,F)\in \mathcal{F}: a_i<j} \bar{v}_{aF}\cdot \tau(i,j,a,F)\right),
 \end{align*}
 where the equality holds because $\bar{z}_{ij}>0$ implies that the corresponding bucket was taken, i.e., that the corresponding dual inequality became tight. Rearranging this sum we obtain
 \begin{align*}
     \sum_{(a,F) \in \mathcal{F}}\bar{v}_{aF}\cdot\sum_{i \in N} \sum_{j=a_i+1}^m  \tau(i,j,a,F) & \leq 2 \sum_{(a,F)\in \mathcal{F}}\bar{v}_{aF}D(a),
 \end{align*}
 where the inequality is a direct application of Lemma~\ref{lm:AproxCompSlackness}. This shows that the cost of the primal solution is at most twice the value of the dual solution. The theorem then follows from weak duality. \qed
\end{proof}

Hence, to show the theorem it suffices to prove Lemma~\ref{lm:AproxCompSlackness}. 
\begin{proof}[Lemma~\ref{lm:AproxCompSlackness}]
 
 Let $(\bar{\imath},\bar{\jmath})$ be the indices of the dual inequality that became tight in the last iteration of the algorithm, and let $\bar{z},\bar{v}$ be the output primal and dual solutions.

Let us fix a variable $v_{aF}>0$ and consider the iteration of the main loop of the algorithm where we were raising that variable.  We split the term to bound in two:

\[\sum_{i\in N}\sum_{j= a_i+1}^m \tau(i,j,a,F)\bar{z}_{ij} = \sum_{j= a_{\bar{\imath}}+1}^m \tau(\bar{\imath},j,a,F)\bar{z}_{\bar{\imath}j} +\sum_{i\in N\setminus\{\bar{\imath}\}}\sum_{j= a_i+1}^m \tau(i,j,a,F)\bar{z}_{ij}.\]
We will bound each term by $D(a)$. For a given item $i\in N$, the expression $\sum_{j\geq 1} \tau(i,j,a,F)\bar{z}_{ij}$ represents the total number of buckets that are receiving water from the external source and whose water is spilling to some bucket that ends up in the final solution.

Regarding item $\bar{\imath}$, notice that  $\sum_{j \geq 1}
 \tau(\bar{\imath},j,a,F)\bar{z}_{\bar{\imath}j}\leq D(a)$, just because the buckets obtaining water from the external source are in the interval $a_{\bar{\imath}}+1,\ldots,m_{\bar{\imath}}(a)$, which are at most $D(a)$ many.% fromwe truncate $\tau$ to exclude those buckets over $a_{\bar{\imath}} + D(a)$, and hence the total amount of water from the external source reaching the buckets of item $i$ in any iteration is less than $D(a)$. 

% To prove Lemma \ref{lm:AproxCompSlackness}, note that for any $i,a,F$, the expression $\sum_{j\geq a_i+1} \tau(i,j,a,F)\bar{z}_{ij}$ represents the total number of buckets that are pouring, for that combination of $a,F$, on the buckets of item $i$ that are part of the final solution but have not been taken yet. With this in mind,  let us divide the inequality into two different ones. On the one hand, regarding the last item: $\sum_{j \geq a_{\bar{\imath}}+1} \bar{z}_{\bar{\imath},j} \tau(\bar{\imath},j,a,F)\leq D(a)$ just because we truncate $\tau$ to exclude those buckets over $a_{\bar{\imath}} + D(a)$. 

Consider now $i\neq \bar{\imath}$ and a bucket $B_{ij}$ that is ``part''of $\tau(i,j',a,F)$ for some segment~$j'$ included in the final solution, that is, either $j'=j$ or $B_{ij}$ is pouring into $B_{ij'}$, case in which all the buckets between $j$ and $j'$ are full in this iteration of the algorithm. Then by construction $B_{ij}$ will be taken as well ($\bar{z}_{ij}=1$). %: otherwise, it would mean that when we took $B_{ij'}$, it must be that $j>m_i(a)$, implying that no more buckets are needed and hence $i= \bar{\imath}$. 
Additionally, no water (either directly or indirectly) reaches a bucket $B_{ik}$ with $k\le a_i$, and hence $\tau(i,k,a,F)=0$.
So the quantity $\sum_{i\in N\setminus \{\bar{\imath}\}}\sum_{j\geq 1} \tau(i,j,a,F)\bar{z}_{ij}$ is upper bounded by the total number of buckets in the final solution, of items other than $\bar{\imath}$, that are above $a$. This number cannot be higher than $D(a)$, otherwise the algorithm would have finished before filling the last bucket $B_{\bar{\imath},\bar{\jmath}}$.\qed %the algorithm returns a solution that satisfies the total demand exactly.\qed

\end{proof}

\textbf{Remark}: %This non-linear optimization problem might be seen as deciding on the following trade-off: when is it convenient to take the first segment(s) of an item, even if they are expensive, due to the promise of cheap high segments? This question, trivially solved in the case of convex or concave functions, becomes complex for the general case. The remarkable characteristic of this algorithm is that not only it emerges from an explicit primal-dual formulation and provides some guaranteed result; it is also a quite intuitive way to face that trade-off. In fact, the water-filling phase of the algorithm gives us a precise procedure, in which having cheap high segments increase the filling rate of the lower segments. Moreover, 
As explained above, the intuition behind the algorithm is that having cheap high segments increase the filling rate of the lower segments. Note that the algorithm does not give an optimum precisely because of this ``promise'': in the last iteration, we are not taking all the full buckets that were pouring to the first non-full bucket; these full buckets that are not taken represent that we are taking into account some future cheap buckets that will not be part of the solution (like unfulfilled promises). Nevertheless, this only happens in the last iteration, which is why the algorithm achieves a 2-approximation.

\subsection{A Rounding Prodecedure}

In this section we show a rounding procedure that takes a solution $z$ of [GKC] (or equivalently [P-GKC]) and returns a feasible solution for the Knapsack-Cover problem with Non-Linear Costs with cost at most twice the cost of $z$. We also give a polynomial time separation algorithm for [GKC], in the list model, by reducing the separation problem to an instance of a problem similar to \NLKC{}. Although this seems circular, as one needs to solve the same problem we are aiming to solve, the separation routine and the rounding technique might be useful for more general problems.

% In this section we show a rounding procedure that takes a solution $z$ of [GKC] (or equivalently [P-GKC]) and returns a feasible solution for the Knapsack-Cover problem with Non-Linear Costs with cost at most twice the cost of $z$. 

We first start with the following simple structural observation about solutions of~[GKC]\footnote{This lemma implies that if in [GKC] we replace each term $\min\left\lbrace z_{ij},z_{i,j-1},\ldots,z_{i1} \right\rbrace $ by $z_{ij}$ and add inequalities  $1\ge {z}_{i1} \ge \ldots \ge {z}_{im}\ge 0$ for all $i\in N$, we obtain an equivalent relaxation. However, the new inequalities yield a different structure of the dual problem, which does not allow for a greedy dual algorithm as in Section~\ref{sec:Primal-DualKC}.}.

\begin{lemma}\label{lm:OptLPStructure}
 Let $z$ be a solution to [GKC]. There exists a polynomial time procedure to create a new solution $\bar{z}$ to [GKC] whose cost is less or equal the cost of $z$ and that satisfies that $1\ge \bar{z}_{i1} \ge \ldots \ge \bar{z}_{im}\ge0$ for all $i\in N$.
\end{lemma}

\begin{proof} 
 %Let us assume that the property of the lemma is not satisfied for $z$ and there exists $r,k$ such that $z_{rk} < z_{r,k+1}$. 
% Let us note that $u_{r,k+1}(a)\le u_{r,k+1}$ for all $a$, and thus
%  
%  \begin{equation}
%   \label{eq:order}
%   \dfrac{z_{rk}}{u_{rk}} < \dfrac{z_{r,k+1}}{u_{r,k+1}} \le \dfrac{z_{r,k+1}}{u_{r,k+1}(a)}.
%  \end{equation}
%  
 For any feasible $z$ we define $\bar{z}_{ij}=\min\{z_{ij},z_{i,j-1},\ldots,z_{i1}\}$ for all $i,j$.  It is straightforward to check that $\bar{z}$ is also feasible, and $\bar{z}\le z$, which implies that its cost did not increase. This directly implies that $\bar{z}_{i1} \ge \ldots \ge \bar{z}_{im}\ge0$ for all $i\in N$.
 
%To show feasibility, let us consider some inequality in~\eqref{eq:GKC-Convex} for some $a$. Any term in the minimum which contains $z_{r,k+1}$ is of the form 
% \[\min\left\lbrace %z_{rj},z_{r,j-1},\ldots,z_{r,k+1},z_{rk}%,\ldots,z_{r1}\right\rbrace,\] for some %$j\ge k+1$, and thus $z_{rk}$ is also %part of the terms in the minimum. %Clearly, changing $z_{r,k+1}$ by %$z_{rk}$ does not change the minimum %attained. Hence, $\bar{z}$ is feasible. 
  
  It is left to argue that we can assume $\bar{z}_{i1}\le 1$. For a given $\ell\in N$, let $k$ be the largest index such that $\bar{z}_{\ell1}\ge \ldots,\bar{z}_{\ell k}>1$. We claim that reassigning those variables a value of 1 implies that $\sum_{ij}\bar{z}_{ij}\ge D$. Indeed, consider $a$ defined as $a_\ell=k$ and $a_i=0$ for $i\neq \ell$. Then 
  \[
   \sum_{i\in N} \sum_{j=a_i+1}^{m} \bar{z}_{ij}
   \ge \sum_{i\in N} \sum_{j=1}^{m_i(a)} \bar{z}_{ij} \ge D(a) \ge D- k,
  \] 
  were the second inequality is implied as $\bar{z}$ is feasible for [GKC]. This implies our claim as the left hand side of this expression equals $\sum_{j=k+1}^{m} \bar{z}_{\ell j}+\sum_{i\in N\setminus\{\ell\}} \sum_{j=1}^{m} \bar{z}_{ij}$. As the new solution, after replacing the variables larger than 1, still satisfies $\bar{z}_{i1} \ge \ldots \ge \bar{z}_{im}\ge0$, then $\bar{z}_{ij}=\min\{\bar{z}_{ij},\bar{z}_{i,j-1},\ldots,\bar{z}_{i1}\}$, we conclude that the new solution satisfies \eqref{eq:GKC-Convex} for $a=0$. The proof of other values of $a$ is analogous.
%   Again, because of~\eqref{eq:order} the minimum is never attained at $\dfrac{z_{r,k+1}}{u_{r,k+1}}u_{ij}(a)$ and changing $z_{r,k+1}$ by $\bar{z}_{r,k+1}$ does not change the contribution to the left hand side of~\eqref{eq:GKC-Convex}.
% %  \end{enumerate}
%  The described procedure constructs a new solution $\bar{z}$ that only differs to $z$ in coordinate $(r,k)$. Repeating this algorithm for all $r\in N$ and for all $k=1,\ldots,m$ (in that order), we obtain the claimed optimal solution.  \qed
\end{proof}

Let us consider an optimal solution $\bar{z}$ of [GKC] satisfying the property of this lemma. %The rounding procedure works in different phases. The first phase, which is applied several times, has the aim of obtaining a new solution, for a reduced number of segments, having two properties: (i) it is a feasible LP solution of a reduced instance where twice of the demand is covered and (ii) it cost at most twice the optimal cost of [GKC] of the respective segments. Such solution can then be rounded by relaxing the generalized knapsack-cover inequalites by a similar condition as of Lemma~\ref{lm:OptLPStructure} and then study the extreme points of such LP. %problem where the demand is that rounds up segments that are taken by $z^*$ more than 1/2 of its capacity and double the other variables. The objective of this is to obtain a new residual problem, that only considers segments taken fractionally, where the fractional variables cover twice the residual demand.
% \paragraph{Phase 1.}
 For the variables such that $\bar{z}_{ij}\ge 1/2$, we can simply round them up to $1$, which increases their contribution to the objective value by at most a factor~2. For the residual problem, the basic idea is to double the other variables and cover twice the residual capacity. %However things are not as simple as the residual problem must consider truncated segments, which destroys the structure that Lemma~1 yields, which is needed for further steps. 
 More precisely, let $\bar{a}_i= \max\{j\in M: \bar{z}_{ij}\ge 1/2\}$ for all $i\in N$. For ease of notation let us call $\bar{D}=D(\bar{a})$ and $\bar{m}_i=m_i(\bar{a})$ for all $i$. By Lemma~\ref{lm:OptLPStructure}, we have that $\bar{z}_{ij}\ge 1/2$ iff $j\le \bar{a}_i$. With this we define the following \emph{residual problem} with duplicate demand, 
 \begin{align}
\nonumber
\text{[R-GKC] }& \min \sum_{i \in I}\sum_{j=\bar{a}_i+1}^{\bar{m}_i} g_{ij}z_{ij}\\
\label{eq:GKC-F}
& \sum_{i\in N}\sum_{j=\bar{a}_i+1}^{\bar{m}_i} z_{ij} \ge 2\bar{D},\\
\label{eq:LPOrder}
& 1\ge z_{i,\bar{a}_i+1} \ge z_{i,\bar{a}_i+2} \ge\ldots \ge z_{i,\bar{m}_i} \ge 0 &\text{for all }i\in N.
\end{align}
 
Notice that $(2\bar{z}_{ij})_{i\in N, j\in \{\bar{a}_i+1,\ldots\bar{m}_i\}}$ is a feasible solution to this problem, as $\bar{z}_{ij}=\min\{\bar{z}_{ij},\bar{z}_{i,j-1},\ldots,\bar{z}_{i1}\}$. Also, we remark that, unlike [GKC], this linear program has only polynomially many constraints. We now show that an optimal extreme point of this program can be easily rounded to a feasible solution to the Knapsack Cover Problem with Non-Linear Cost. Let $z^*$ be an optimal extreme point solution to [R-GKC]

\begin{lemma} Consider an optimal extreme solution $z^*$ to [R-GKC]. Then there exists at most one item $i \in N\times M$ such that $z^{*}_{ij}\in (0,1)$ for some $j$. All other items $k\in N\setminus{i}$ satisfy that $z^*_{kj}\in \{0,1\}$ for all $j$.
% vector $(w_i)_{i\in N}\in [0,1]^N$ and numbers $r_i,R_i$ where $\bar{a}_i< r_i\le R_i\le m_i$, such that for all $i\in N$, 
% \[
%  \frac{\hat{z}_{ij}}{\bar{u}_{ij}} = \begin{cases}
%                   1 & \text{for }\bar{a}_i< j < r_i, \\
%                   w_i &\text{for }r_i\le j \le R_i,\\
%                   0   & \text{for }R_i< j \le m_i.\\
%                 \end{cases}
% \]
% Moreover, there exists at most one $i\in N$ with a fractional value $w_i\in (0,1)$.
\end{lemma}
\begin{proof}
 Let us fix a particular item $i \in N$ and let us consider the LP restricted to the corresponding vector $z_i = (z_{ij})_{j=\bar{a}_i+1}^m$, that is, we consider the LP:
 \begin{align}
\nonumber
\text{[R-GKC]}_i& \min \sum_{j=\bar{a}_i+1}^{\bar{m}_i} g_{ij}z_{ij}\\
\label{eq:R-GKC-F}
& \sum_{j=\bar{a}_i+1}^{\bar{m}_i} z_{ij} = W_i,\\
\label{eq:R-LPOrder}
& 1\ge z_{i,\bar{a}_i+1}\ge z_{i,\bar{a}_i+2} \ge\ldots \ge z_{i,\bar{m}_i} \ge 0,
\end{align}
 where $W_i= 2\bar{D} - \sum_{k\neq i}\sum_{j=\bar{a}_k+1}^{\bar{m}_k} z^*_{kj}$. 
 Clearly $z^*_i=(z^*_{ij})_j$ is optimal for this problem, and must also be an extreme point, as if it can be written as a convex combination of two different vectors then also $z^*$ can. As the LP has $s_i = \bar{m}_i-\bar{a}_i$ variables, then $s_i$ inequalities must be satisfied with equality, which implies that at least $s_i-1$ inequalities in~\eqref{eq:R-LPOrder} must be satisfied with equality. This implies that there exist parameters $\bar{a}_i+1\le r_i\le R_i\le \bar{m}_i$ such that $z^*_{ij}=1$ if $j<r_i$, 
 \[z^*_{i,r_i}=z^*_{i,r_i+1}=\ldots = z^*_{i,R_i},\]
 and $z^*_{ij}=0$ for $j>R_i$. Let us set $w_i=z^*_{i,r_i}$.
 
 To see that there is at most one fractional variable $w_i$, let us consider the following LP,
\begin{align} \textnormal{[S-GKC] } & \sum_{i\in N} \omega_i \bar{g}_i \\
& \sum_{i\in N} \omega_i \ge 2\bar{D}-\sum_{i\in N}(r_i-\bar{a}_i-1),\\
& 0 \le \omega_i\le 1 & \text{for all } i\in N. 
\end{align}
where $\bar{g}_i = \sum_{j=r_i}^{R_i} g_{ij}$. It is not hard to see that $(w_i)$ must be an optimal extreme solution to this LP, as otherwise $z^*$ would not be an optimal extreme solution to [R-GKC]. The lemma follows. As [S-GKC] has only one inequality besides $ w_i\in [0,1]$ for all $i$, we conclude that in a extreme optimal solution there is at most one fractional variable.\qed
\end{proof}

Finally, let $k\in N$ be the unique item with fractional variables in $z^*$ (if any). We will round to zero the variables $z^*_{kj}$ for this item. Putting the pieces together we show that the following solution is a 2-approximation

\[\hat{z}_{ij} = \begin{cases}
           1 & \text{for all }i\in N, j\in\{1,\ldots,\bar{a}_i\}, \\
           z^*_{ij} & \text{for all }i\in N\setminus\{k\}, j\in\{\bar{a}_i+1,\ldots,\bar{m}_i\},\\
           0 & \text{otherwise}.
          \end{cases}
\]

\begin{theorem}
 The construction solution $\hat{z}$ is an integral feasible solution to [GKC] that satisfies that $\sum_{ij}g_{ij}\hat{z}_{ij}\le 2\sum_{ij}g_{ij}\bar{z}_{ij}$ and thus it is a 2-approximate solution. 
\end{theorem}
\begin{proof}
 The fact that $\hat{z}$ is integral follows directly by construction as the unique fractional coordinate from $z^*$ is rounded down to zero in $\hat{z}$. To show feasibility first note that the construction satisfies directly that $\hat{z}_{ij}=1$ implies that $\hat{z}_{i,j-1}=1$. Moreover, the demand is covered as
 \begin{align*}
  \sum_{i\in N} \sum_{j\in M}\hat{z}_{ij} &= \sum_{i\in N}\left(\bar{a}_i +\sum_{j=\bar{a}_i+1}^{\bar{m}_i} z^*_{ij}\right) - \sum_{j=\bar{a}_k+1}^{\bar{m}_k} z^*_{kj} \\
  & \ge \sum_{i\in N}\left(\bar{a}_i +\sum_{j=\bar{a}_i+1}^{\bar{m}_i} z^*_{ij}\right) - \bar{D} \\
  & \ge \sum_{i\in N}\bar{a}_i + \left(D-\sum_{i\in N}\bar{a}_i\right)_+ \ge D,
 \end{align*}
 where the first inequality follows as $\bar{m}_k -\bar{a}_k= m_k(\bar{a})- \bar{a}_k \le D(\bar{a})=\bar{D}$, and the second inequality since $z^*$ is feasible for [R-GKC]. Finally, to show the approximation ratio notice that 
 \begin{align*}
 \sum_{i\in N}\sum_{j\in M}g_{ij}\hat{z}_{ij} &\le \sum_{i\in N}\left(\sum_{j=1}^{\bar{a}_i}2g_{ij}\bar{z}_{ij} + \sum_{j=\bar{a}_i+1}^{\bar{m}_i}g_{ij}z^*_{ij}  \right)\\
 &\le \sum_{i\in N}\left(\sum_{j=1}^{\bar{a}_i}2g_{ij}\bar{z}_{ij} + \sum_{j=\bar{a}_i+1}^{\bar{m}_i}2g_{ij}\bar{z}_{ij}  \right)\\
 &= 2\sum_{i\in N}\sum_{j\in M}g_{ij}\bar{z}_{ij},
 \end{align*}
  where the first inequality follows from the definition of $\hat{z}$, and the second from the fact that $(2\bar{z}_{ij})_{i\in N, j\in \{\bar{a}_i+1,\ldots\bar{m}_i\}}$ is feasible for [R-GKC].\qed
\end{proof}

To finish this section we show a polynomial time separation algorithm for [GKC], again for the list model. Let $z$ be a feasible solution to this program. By Lemma~\ref{lm:OptLPStructure}, we can assume that $0\le z_{ij} = \min\{z_{ij},z_{i,j-1},\ldots,z_{i1}\}$ for all~$i$. As the instance is feasible we have that $D\le nm$, and hence it is polynomially bounded. We split the separation problem for each value of $D(a)$. Let $d\le D$ be an integer which we fix from now on. Hence, it suffices to minimize $\sum_{i\in N}\sum_{j=a_i+1}^{m_i(a)} z_{ij}$ over all possible $a\in \{0,\ldots,m\}^N$ where $D(a)= d$. Recall also that $m_i(a)=\min\{m,a_i+D(a)\}$, and hence this term equals $\min\{m,a_i+d\}$ and dependes only on variable $a_i$. Let $h_i(a_i)=\sum_{j=a_i+1}^{m_i(a)} z_{ij}$. We can reinterprate the problem to solve as $\min \sum_{i\in N}h_i(a_i)$ subject to $\sum_{i\in N} a_i = D-d$. This is a problem similar to \NLKC{}, with the difference that the demand must be covered exactly and the cost functions are not non-decreasing.

We can easily solve this problem with a dynamic program by creating a table $T$ with entries $T(r,e)$ which denote the minimum value of $\sum_{i=1}^{r}h_i(a_i)$ achievable over all values of $(a_i)_i^r$ such that $\sum_{r=1}^i a_i = e$. The optimal value is given by $T(n,D-d)$. The table can be filled in polynomial time (in the list model) by noting that
\[
T(r,e) = \min\{T(r-1,e-a_i) + h_i(a_i):a_i\in\{0,\ldots,m\}\}.
\]

\section{Unsplittable Flow-Cover on the Line } \label{sec:UFP}

% We now consider this non-linear approach over the Unsplittable Flow Problem Cover (UFP-cover), that extends KP. In the original UFP-cover, we have a discrete interval $I=\{1,...,k\}$ and a set $N$ of $n$ items, each one characterized by a capacity or height $u_i$, a cost $c_i$ and a sub-interval $I_i\subseteq \{1,\ldots,k\}$. We also have a demand $D_t$ for each $t \in I$. The problem consists on selecting the cheapest set of items such that the total height at any point it at least the demand, that is, we must pick a the set $S$ minimizing $\sum_{i\in S}c_i$ such that $\sum_{i\in S:I_i\ni t} u_i\ge D_t$ for all $t\in I$. For this classic version of the problem, there is an algorithm that provides a 4-approximation.
We now show that extending the ideas of Section~\ref{Sec KP} we can also achieve a 4-approximation for the \NLUFP{} problem. Recall that an instance of this problem is given by an interval $I=\{1,...,k\}$, a set $N$ of $n$ items, where eacy item is characterized by a capacity or height $u_i$, a cost $c_i$, and a sub-interval $I_i\subseteq \{1,\ldots,k\}$. We also have a demand $D_t$ for each $t \in I$.

%In the non-linear case, we can choose the height of each item from within a set $M=\{1,\ldots,m\}$. In other words, each item consists of a list of (vertical) segments, and one must choose a prefix of them. Again, we first focus on the list model where the costs of segments $g_{i1},g_{i2},\ldots,g_{im}$ are given in a list. %Notice that without loss of generality we can assume $k\le 2n$ as a point in $I$ which is not an end-point of some $I_t$ can be easily removed, thus the length of $I$ can be assumed to be of polynomial size. 
%As before, we use variables $z_{ij}\in\{0,1\}$ to represent if segment $j$ of item $i$ is in the solution. With this the cost  is $\sum_{i\in N}\sum_{j\in M} z_{ij}g_{ij}$ and the constraint that each demand must be covered is given by 
%\begin{equation}\label{eq:CoverUFP}\sum_{i\in N: t \in I_i}\sum_{j\in M} z_{ij}\geq D_t \text{ for all }t\in I.\end{equation}
%Finally, we require that $z_{i,j+1}=1 \Rightarrow z_{ij}=1$ for all $i\in N,j\in\{1,\ldots,m-1\}.$
In the non-linear case, we can choose the height of each item from within a set $M=\{1,\ldots,m\}$. In other words, each item consists of a list of (vertical) segments, and one must choose a prefix of them. The costs of the segments for item $i$ are given by $g_{i1},g_{i2},\ldots,g_{im}$. Notice that without loss of generality we can assume $k\le 2n$ as a point in $I$ which is not an end-point of some $I_t$ can be easily removed, thus the length of $I$ can be assumed to be of polynomial size. An exact formulation of this problem is the following
\begin{align}
\nonumber
  \min &    \sum_{i\in N}\sum_{j\in M} z_{ij}g_{ij} \\
\label{eq:CoverUFP}\text{s.t. } & \sum_{i\in N: t \in I_i}\sum_{j\in M} z_{ij}\geq D_t & \forall t \in I,\\
\nonumber & z_{i,j+1}>0 \Rightarrow z_{ij}=1 & \forall i\in N,j\in\{1,\ldots,m-1\},\\
\nonumber & z_{ij} \in \{0,1\} & \forall i\in N,j\in M. 
\end{align}
% \begin{align}
% \nonumber
%   \min &    \sum_{i\in N}\sum_{j\in M} z_{ij}g_{ij} \\
% \label{eq:CoverUFP}\text{s.t. } & \\
% \nonumber & z_{i,j+1}=1 \Rightarrow z_{ij}=1 & \forall i\in N,j\in\{1,\ldots,m-1\},\\
% \nonumber & z_{ij} \in \{0,1\} & \forall i\in N,j\in M. 
% \end{align}
% This problem assumes that all segments are of unit length, and hence the formulation is of pseudopolynomial size. We show in Section~\ref{sec:Non-unitSegments} how to turn our algorithms to run in polynomial time.\afcom{Este parrafo es redundante con como quedo el paper.}

We now present the problem using the generalized knapsack-cover inequalities, and the relaxation explained in Section \ref{Sec KP}, applied to each inequality in~\eqref{eq:CoverUFP} separately. For this, define
% \begin{equation*}
    $D_t(a)=\max \left(D_t - \sum_{i: t \in I_i}a_i,0\right)$
% \end{equation*}
and $
% \be$gin{equation*}
    \tau (i,j,a,F,t)=\{k\ge j: k \in F_i^{k-j}, k \leq m_i(a)\},$
% \end{equation*}
where $m_i(a)=\min\{m,a_i+D_t(a)\}$. %The relaxed primal problem, which we call [P-UFP], asks to minimize $\sum_{i,j} z_{ij}g_{ij}$ for $z\ge0$ subject to
%\begin{equation*}
    % \text{[P-UFP]:} \min & \\
 %   \sum_{i: t \in I_i}\sum_{j \ge 1} \tau (i,j,a,F,t)\cdot z_{ij} \geq D_t(a) \quad \text{for all } (t,a,F) \in \mathcal{H}
%\end{equation*}
%where $\mathcal{H}$ is the set of triplets $(t,a,F)$ where $t\in I$ and $(a,F)\in \mathcal{F}$, as defined in Section~\ref{sec:Primal-DualKC}. %:\substack{\displaystyle t \in I, a\in M^N,\, \bigcupdot_{k=0}^{m-1}F_i^k= \{a_i+1,\ldots,m\},\\\displaystyle\text{and if }  j\in F_i^k \text{ then } j-k\ge 1 \text{ for all } k, j} 
% \right\rbrace.
% $
% \jvcom{Donde vive $t,a,F$?}
% \afcom{Agregado. Lo puse sólo al comienzo porque la notación posterior ya es demasiado pesada}
The relaxed primal problem is:
\begin{align*}
    \text{[P-UFP]:} \min & \sum_{i,j} z_{ij}g_{ij}\\
    \text{s.t. } & \sum_{i: t \in I_i}\sum_{j \geq a_{i+1}} \tau (i,j,a,F,t)\cdot z_{ij} \geq D_t(a) & \forall (t,a,F) \in \mathcal{H},\\
    & z \geq 0,
\end{align*}
where
\[
 \mathcal{H}=\left\lbrace (t,a,F):\substack{\displaystyle t \in I, a\in \{0,\ldots,m\}^N,\, \bigcupdot_{k=0}^{m-1}F_i^k= \{a_i+1,\ldots,m\},\\\displaystyle\text{and if }  j\in F_i^k \text{ then } j-k\ge 1 \text{ for all } k, j} 
\right\rbrace.
\]
%This yields a dual relaxation %, which we denote 
%[D-UFP], whose objective is $\max  \sum_{(a,t,F)\in \mathcal{H}} v_{atF}\cdot D_t(a)$, and we must optimize over all $v\ge0$ satisfying
%\begin{equation}
    % \text{[D-UFP]:}  \\
  %  \sum_{(a,t,F)\in \mathcal{H}:t \in I_i}
   % \tau(i,j,a,F,t)\cdot v_{atF} \leq g_{ij} \quad \text{ for all } i\in N,j\in M.
    % & v \geq 0.
%\end{equation}
This yields the following dual \\
\begin{align*}
    \text{[D-UFP]:} \max & \sum_{(a,t,F)\in \mathcal{H}} v_{atF}\cdot D_t(a) \\
    \text{s.t. } & \sum_{t \in I_i}\sum_{a: j\geq a_{i+1}} \tau(i,j,a,F,t)\cdot v_{atF} \leq g_{ij} & \forall i\in N,j\in M,\\
    & v \geq 0.
\end{align*}
\noindent\textbf{Algorithm Description.} %Our primal-dual algorithm is given in Algorithm~1. In this case our approach has two phases. During the \emph{growing} phase (Lines 5--15), we construct a dual solution, which then directly implies a feasible primal solution. In the \emph{pruning} phase (Lines 16--20) we remove unnecessary segments from the primal solution. 
We now  show Algorithm \ref{alg:UFP}, which is the result of applying the buckets ideas of Section~\ref{Sec KP} to the 4-approximation algorithm for UFP-cover~\cite{bar-noy_unified_2001}.

%\afcom{Falta la cita}%\jvcom{Falta todavía precisar/detallar mejor el algoritmo en el texto}

As many primal-dual algorithms our approach has two phases. During the \emph{growing} phase, we construct a dual solution, which then directly implies a feasible primal solution. In the \emph{pruning} phase we remove unnecessary segments from the primal solution.As before, for each item $i\in N$ we have a stair of buckets, where each bucket $B_{ij}$ corresponds to a given inequality in the dual, indexed by $j\in M$ and $i\in N$. In each iteration of the growing phase buckets receive water (that might fall to inferior buckets) from an external source at a rate of 1 or 0. Once we define the rates, the water dynamics work in exactly the same way as in Section~\ref{sec:Primal-DualKC}: water reaching a given bucket that is full is spilled to the next bucket to the left until it reaches a bucket that is not full. The only difference is that only some of the items receive water. More precisely, in a given iteration of the growing phase, we select $t\in I$ with largest unsatisfied demand $D_t(a)$ (break ties arbitrarily). This is a greedy criterion to increase the dual objective function as fast as possible. Only buckets for items $i\in I$ such that $t\in I_i$ receive water from the external source. For such an item $i$, the subset of buckets receiving water from the external source are again buckets $B_{ij}$ with $j\in\{a_i+1,\ldots,m_i(a)\}$. The water dynamics can be emulated by raising a single dual variable at a time. Notice that the only difference to the dual in Section~\ref{Sec KP} is that when raising a given variable $v_{atF}$, only inequalities for items $i\in N$ where $t\in I_i$ are affected, corresponding to the fact that only buckets corresponding to such items receive water from the external source.

When one or more buckets become full, we pick one of these buckets. As before, a full bucket means that the corresponding segment $(i,j)$ is available. We take a given segment, that is, we define a primal variable $z_{ij}$ to 1, as soon as all preceding buckets of item $i$ are available. In other words, if $B_{ij}$ becomes full for $j=a_i+1$, then we take set $z_{ij}=\ldots=z_{iq}=1$ where $B_{ij},\ldots,B_{iq}$ are full but $B_{i,q+1}$ is not. This is the case even if $q>m_i(a)$. All taken buckets (or segments) are considered to be a ``block'', denoted by $b_{ij}$ (where $j$ denotes the first segment of the block). After this we update $t$ and continue with a new iteration of the main loop of the growing phase.

Although this first phase gives a feasible primal solution, some blocks might have become redundant, that is, the solution would remain feasible without them. In the second phase we remove redundant blocks when we can. To do this, we check for each block $b_{ij}$, in reverse order in which they were added, whether removing the block from the primal solution makes the given primal unfeasible. If $b_{ij}$ is redundant and it is the superior block (i.e., the block  containing the highest segment that is still in the solution) of its item, we remove it; if $b_{ij}$ is redundant but there are blocks over it in the solution when it is checked, we cannot remove it (the solution would become unfeasible). Note that doing so, all the superior blocks that are in the final solution are not redundant.

 \begin{algorithm}
\noindent\textbf{Algorithm 3.}
Primal-Dual Algorithm for \NLUFP{}.

\label{alg:UFP}
%\vspace{-0.2cm}
%\vspace{-0.4cm}
\noindent\rule{\textwidth}{0.01cm}
\begin{algorithmic}[1]

\STATE $z$, $v$ $\leftarrow$ 0 \Comment{primal and dual solutions.}
\STATE  $a \leftarrow 0$ 
\STATE $F_i^k\leftarrow \emptyset$  for all  $i,k\ge 1$; $F_i^0 \leftarrow M$
% \STATE 
\STATE $\mathcal{B}_i \leftarrow \emptyset$ for all $i$ \Comment{Set containing the blocks of item $i$.}

% \Comment{Growing Phase.}
\WHILE{ $D_t(a)>0$ for some $t$}
\STATE Select $t$ that maximizes $D_t(a)$ (break ties arbitrarily). \label{ln:mainLoopAlg3}
\STATE Increase $v_{atF}$ until a dual constraint indexed by $(i,j)$, for some item $i$ and segment $j$, becomes tight. Break ties in favor of a bucket with smallest index~$j$.
\IF{$j>a_i+1$}
    \STATE $F \leftarrow \text{Update}(F,i,j,a)$ \Comment{Water pouring into $B_{ij}$ pours into a lower bucket.}
\ELSE
	\STATE Let $q>a_i$ be the maximum number such that $j' \notin F_i^0$ for all $ j'=j+1,\ldots,q$ \Comment{we take all full buckets that poured into $B_{i,a_i+1}$, even the ones that are now truncated.}\label{eq:takeBlock} 
	\STATE Set $a_i \leftarrow q$ 
	\State Set $\mathcal{B}_i \leftarrow \mathcal{B}_i \cup b_{ij}$, with $b_{ij}= \{j,\ldots,q\}$ \Comment{$b_{ij}$ is a block that enters the primal solution.}
	
\ENDIF
\ENDWHILE

%Phase 2
% \Comment{Pruning phase}
\FORALL{ $b_{ij}$ in reversed order in which they are defined in the growing phase}
\IF {$b_{ij}$ can be removed from the primal solution without leaving any demand unsatisfied \textbf{and} $j\geq j'$ for all $j'$ such that $b_{ij'} \in \mathcal{B}_i$}

\Comment{We eliminate redundant blocks, unless they have a superior block over it}

\STATE $\mathcal{B}_i \leftarrow \mathcal{B}_i \setminus b_{i,j}$
\ENDIF

\ENDFOR
\STATE \textbf{return} $z$, where $z_{ij}=1$ for $j\le \sum_{b_{ik}\in \mathcal{B}_i} |b_{ik}|$.
\end{algorithmic}
\vspace{-0.3cm}
\noindent\rule{\textwidth}{0.04cm}
 \end{algorithm}
% \medskip
% The algorithm finishes since for each iteration of the growing phase their is a bucket filled. Also the algorithm maintains throughout its execution a feasible dual solution, and the growing phase finishes only when a feasible primal solution is found. By construction the pruning phase does not change feasibility. 

%The proof of correctness is analogous as in Section~\ref{sec:Primal-DualKC}. To show the approximation factor we need the following key lemma, whose proof is in the appendix. Again, theorem \ref{thm:ArbitraryFunctions2} requires a simple adaptation of the algorithm.

\noindent\textbf{Analysis.} The algorithm finishes since for each iteration of the growing phase their is a bucket filled. Also the algorithm maintains throughout its execution a feasible dual solution, and the growing phase finishes only when a feasible primal solution is found. By construction the pruning phase does not change feasibility. To prove that this is a 4-approximation, we just need to prove Lemma \ref{LemmaHolguraUFP}. The rest of the proof relies on the usual primal-dual techniques, equivalent to the use of Lemma \ref{lm:AproxCompSlackness} to prove Theorem \ref{thm:PrimalDual2}.

\begin{lemma}
\label{LemmaHolguraUFP}
Consider the output $(z,v)$ of the algorithm. Let $(a,t,F)$ such that $v_{atF}>0$. Then \[\sum_{i: t\in I_i} \sum_{j \ge 1} \tau(i,j,a,F,t) \cdot z_{ij}\leq 4 D_t(a).\]
\end{lemma}

\begin{proof}
 Let us fix a variable $v_{atF}>0$ raised in the main loop of the growing phase. Denote by $\mathcal{B}_i$ the set of blocks for each $i$ at the end of the pruning phase. 
 Out of those, consider the ones that are above $a$, and that contribute to fulfill the demand in $t$, that is, $S_{ta}=\{b_{ij} \in \cup_{i\in N}\mathcal{B}_i: t \in I_i, j\geq a_{i}+1\}$. Let us denote by $\overline{b}_i$ the superior block of each item (i.e. $\overline{b}_i=b_{ij}$ with $b_{ij} \in S_{ta}$ and $j\geq j'$ for all $j'$ such that $b_{ij'} \in \mathcal{B}_i$). For each of these superior blocks, as they were not removed, it must exist some $t_i \in I$ such that its demand would become unsatisfied when removing $\overline{b}_i$, which is of course also true if we look only at the blocks in $S_{ta}$, i.e.,
% \begin{equation}
% \label{EqDt1}
% D_{t_i}(a)  > \sum_{k \in N: t_i \in I_k} \sum_{\substack{j\geq a_{k+1}:\\ b_{kj} \in \mathcal{B}_k}} |b_{kj}| - |\overline{b}_i|.
% \end{equation}
  \begin{equation}
 \label{EqDt1}
 D_{t_i}(a)  > \left(\sum_{b_{kj} \in S_{ta}} |b_{kj}|\right) - |\overline{b}_i|.
 \end{equation}
 
 %\jvcom{le cambié el signo antes de %$|\overline{b}_i|$}\jvcom{En el subdindice de la segunda %sumatoria debería decir $i'$ en vez de $i$, no? Mejor usar %$k$ en vez de $i'$. También la suma debería considerar %solo items $i'$ tal que $t_i\in I_{i'}$} \afcom{Tienes %razón. Pero creo que queda mejor como lo puse ahora. Dejé %comentado arriba el original por si no estás de acuerdo.}

 Inequality \eqref{EqDt1} is true because we removed the blocks in a reversed order, and the blocks that conform $a$, some of which might have been removed, were introduced before $\overline{b}_i$ in the growing phase. Let us classify the blocks in $S_{ta}$ into two subsets,
$$S_{ta}^L=\{b_{i\ell}\in S_{ta}: t_i\leq t\} \quad \text{ and } \quad  S_{ta}^R=\{b_{i\ell}\in S_{ta}: t_i> t\}.$$

We divide the proof of Lemma \ref{LemmaHolguraUFP} into two analogous inequalities. Let us show that 

\begin{equation}
    \sum_{i: t\in I_i} \sum_{j: b_{ij} \in S_{ta}^R} z_{ij}\tau(i,a,t,F) \leq 2 D_t(a).\label{eq:rightside}
\end{equation}
To do this, define $t^R=\min \{t_i: \overline{b}_i \in S_{ta}^R\}$. Note that $t^R$ is covered by every interval~$I_{i}$ with $\overline{b}_{i}$ in $S_{ta}^R$, as they cover $t$ (which is at most $t^R$) and their $t_i$ (which is larger or equal than $t^R$). Define $(i_1,j_1)$ such that $t_{i_1}=t^R$ and $\overline{b}_{i_1}=b_{i_1,j_1}$. On the one hand, by definition of $\tau$:

\begin{equation}\label{eq:SingleBlock} z_{i_1,j_1}\tau(i_1,j_1,a,t,F)\le\tau(i_1,j_1,a,t,F)\leq D_t(a).\end{equation}
On the other hand we study

$$\sum_{i: t\in I_i} \sum_{\substack{j: B_{ij} \in S_{ta}^R,\\ (i,j)\neq (i_1,j_1)}} z_{ij}\tau(i,j,a,F,t).$$

%Consider an item $i\neq i_1$, and the iteration while increasing variable $v_{atF}$. The summands are the number of buckets that were spilling over each of the segments above $a_i$ that are in the final solution (because we only sum when $z_{ij}=1$). This quantity cannot be higher than the sum of the cardinality of all the blocks above $a_i$ in the final solution (recall that when blocks are taken in Line~\ref{eq:takeBlock}, they include truncated buckets that have poured onto the taken segments). For $i_1$, the same argument holds, but the superior block $\overline{b}_{i_1}$ does not need to be considered because it never poured onto the inferior blocks (otherwise they would have been the same block). Thus it holds that

%\begin{equation}
%  \sum_{i: t\in I_i} \sum_{\substack{j \geq a_i+1:\\ b_{ij} \in S_{ta}^R,\\ (i,j)\neq (i_1,j_1)}} z_{ij}\tau(i,j,a,t,F)  \leq  \sum_{\substack{b_{ij} \in S^R_{ta},\\ (i,j)\neq (i_1,j_1)}} |b_{ij}| \leq D_{t^R}(a) \leq D_t(a).\label{eq:nonSingleBlock}
%\end{equation}

Consider an item $i\neq i_1$, and the iteration while increasing variable $v_{atF}$. The summands are the number of buckets that were spilling over each of the segments above $a_i$ that are in the final solution (because we only sum when $z_{ij}=1$, and buckets $B_{ik}$ for $k\le a_i$ do not receive water). This quantity cannot be higher than the sum of the cardinality of all the blocks above $a_i$ in the final solution (recall that when blocks are taken in Line~\ref{eq:takeBlock}, they include truncated buckets that have poured onto the taken segments). For $i_1$, the same argument holds, but the superior block $\overline{b}_{i_1}$ does not need to be considered because it never poured onto the inferior blocks (otherwise they would have been the same block). Thus it holds that

\begin{equation}
  \sum_{i: t\in I_i} \sum_{\substack{j: b_{ij} \in S_{ta}^R,\\ (i,j)\neq (i_1,j_1)}} z_{ij}\tau(i,j,a,F,t)  \leq  \sum_{\substack{b_{ij} \in S^R_{ta},\\ (i,j)\neq (i_1,j_1)}} |b_{ij}| \leq D_{t^R}(a) \leq D_t(a).\label{eq:nonSingleBlock}
\end{equation}

The second inequality is given by \eqref{EqDt1}, recalling that all the blocks in $S_{ta}^R$ cover $t^R$; the third inequality is due to the greedy criterion to select $t$ as the one maximizing $D_t(a)$ in Line~\ref{ln:mainLoopAlg3}. We conclude~\eqref{eq:rightside} by adding~\eqref{eq:SingleBlock} and~\eqref{eq:nonSingleBlock}. The lemma follows by treating the set $S_{ta}^L$ analogously, and adding~\eqref{eq:rightside} with the analogous expression for $S_{ta}^L$.\qed
\end{proof}

With this lemma we can show the following result. Its proof is completely analogous to the proof of Theorem~\ref{thm:PrimalDual2} and thus it is omitted.

\begin{theorem}
 Algorithm~\ref{alg:UFP} is a 4-approximation algorithm for the UFP-Cover problem with Non-Linear Costs.
\end{theorem}

\section{Arbitrary non-decreasing functions}\label{sec:Non-unitSegments}

We now show how to adapt our algorithms in Sections~\ref{Sec KP} and~\ref{sec:UFP} for arbitrary non-decreasing functions $f_i:\{1,\ldots,m\}\rightarrow \mathbb{Q}_{\ge0}\cup\{\infty\}$, for a given $m$ (not necessarily polynomially bounded) by only losing a factor of $1+\varepsilon$ in the approximation guarantee. We assume that each function $f_i$ is given by an oracle, such that a polynomial number of bits is enough to describe all values $f_i(x)$. Using standard techniques, first we show to approximate each function $f_i$ by a piece-wise constant function with polynomial number of steps. After, we discuss how to emulate Algorithms~\ref{alg:Primal-DualKC} and~\ref{alg:UFP} in polynomial time for such functions.

First of all, by scaling we can assume, without loss of generality, that~\mbox{$f_i(x)\in \mathbb{N}$}.  

\begin{lemma} Consider $\varepsilon>0$ and let $f:\{1,\ldots,m\}\rightarrow \mathbb{N}$ be an arbitrary non-decreasing function. There exists a piece-wise constant function $\tilde{f}$ such that $$f(x)\le \tilde{f}(x)\le (1+\varepsilon)f(x)$$ for all $x\in \{1,\ldots,m\}$, where $\tilde{f}$ has at most $\lceil \log_{1+\varepsilon} f(m)\rceil$ many pieces.
\end{lemma}
\begin{proof}
To prove the lemma we can assume that $f(x)>0$, as the values $f(x)=0$ just corresponds to separate piece in $\tilde{f}$. For all other $x\in\{1,\ldots,m\}$, we can simply set $\tilde{f}(x) = (1+\varepsilon)^{\lfloor\log_{1+\varepsilon}f(x)\rfloor}$. The constant-wise pieces (intervales) of $\tilde{f}$ can be easily computed in polynomial time with a binary search approach.
\qed
\end{proof}

We explain now how to adapt Algorithms \ref{alg:Primal-DualKC} and \ref{alg:UFP} for this scenario. %Notice that if we apply directly the algorithms, they would become pseudo-polynomial \afcom{borraria el notice that... porque es redundante con algo que ponemos mas abajo, donde queda mas claro}. 
Let us partition the set $\{1,\ldots,m\}$ in intervals $J_{i1},J_{i2},\ldots,J_{im_i}$ correspondent to the piece-wise constant pieces of $\tilde{f}_i$. We denote by $u_{ik}\in \mathbb{N}$ the cardinality of interval $J_{ik}\subseteq{\mathbb{N}}$. To adapt the algorithms, note that as they deal with unitary segments, a piecewise constant function can be replaced (preserving the same costs for any solution) by a piecewise constant function with a pseudopolynomial number of segments. More precisely, if $J_{ik}=\{\ell,\ell+1,\ldots,u\}$, then $g_{i\ell} = \tilde{f}_{i}(\ell)-\tilde{f}_{i}(\ell-1)$, and $g_{ir}=0$ for all $r\in\{\ell+1,\ldots,u\}$. Applying our algorithms to this instance would imply a pseudopolynomial running time. However, as all but $m_i$ many buckets for item $i$ has zero capacity $g_{ij}$, we can handle all of them simultaneously to make our algorithms run in polynomial time.

To do this, we can process all segments in $J_{ik}$ in a single step: when the algorithm begins, all their respective buckets but the first one would be full, so the other buckets will receive water at a rate equal to the length of the constant interval. Equivalently, the interval $J_{ij}$ is represented by a bucket of height $g_{ij}$ that gets filled at a rate $u_{ij}$. Any time a bucket pours onto some inferior bucket, its rate also increases by the length of the interval corresponding to the pouring bucket. Truncations, given by the fact that in a given iteration only buckets $B_{ij}$ for $j\in \{a_i+1,\ldots,m_i(a)\}$ get water from the external source for each item $i$, make these rates diminish accordingly. With these rules, the algorithms can be easily adapted to run in polynomial time implying the following theorems:

\begin{theorem}\label{thm:ArbitraryFunctions1}
The adapted version of Algorithm~\ref{alg:Primal-DualKC} is a polynomial time $(2+\varepsilon)$-approximation for the Knapsack-Cover Problem with Non-Linear Costs and arbitrary non-decreasing functions.
\end{theorem}

\begin{theorem}\label{thm:ArbitraryFunctions2}
The adapted version of Algorithm~\ref{alg:UFP} is a $(4+\varepsilon)$-approximation for the Unsplittable Flow-Cover on the Line Problem with Non-Linear Costs and arbitrary non-decreasing functions.
\end{theorem}

% \newpage

\bibliography{bib}
\bibliographystyle{abbrv}

\newpage

\end{document}